\documentclass[submission,copyright,creativecommons]{eptcs}
\providecommand{\event}{TARK 2015} % Name of the event you are submitting to
\usepackage{breakurl}
\usepackage{latexsym}
\usepackage{mathrsfs}
\usepackage{graphicx}
\usepackage{amssymb}
\usepackage{enumerate}
\usepackage[linesnumbered,ruled,vlined]{algorithm2e}
\usepackage{color}
\usepackage{multirow}
\usepackage{amsmath}
\usepackage{cases}
\usepackage[T1]{fontenc}

\newtheorem{theorem}{Theorem}[section]
\newtheorem{lemma}[theorem]{Lemma}
\newtheorem{proposition}[theorem]{Proposition}
\newtheorem{corollary}[theorem]{Corollary}
\newtheorem{fact}[theorem]{Fact}
\newtheorem{proof}[theorem]{Proof}
\newtheorem{definition}{Definition}[section]
\newtheorem{example}{Example}[section]

\title{Preference at First Sight}
\author{Chanjuan Liu
       \institute{School of Computer Science and
                 Technology, Dalian University of Technology}
       \institute{Institute for Logic, Language and Computation, University of Amsterdam}
       \email{chanjuan.pkucs@gmail.com}
}
\def\titlerunning{Preference at First Sight}
\def\authorrunning{C. Liu}

\begin{document}

\maketitle

\begin{abstract} We consider decision-making and game scenarios in which an agent is limited by his/her computational ability to foresee all the available moves towards the future -- that is, we study scenarios with \emph{short sight}. We focus on how short sight affects the logical properties of decision making in multi-agent settings. We start with  \emph{single-agent sequential decision making} (SSDM) processes, modeling them by a new structure of `\emph{preference-sight trees}'.  Using this model, we first explore the relation between a new natural solution concept of \emph{Sight-Compatible Backward Induction}  (SCBI) and the histories produced by classical Backward Induction (BI). In particular, we find necessary and sufficient conditions for the two analyses to be equivalent. Next, we study whether larger sight always contributes to better outcomes. Then we develop a simple logical special-purpose language to formally express some key properties of our preference-sight models. Lastly, we show how short-sight SSDM scenarios call for substantial enrichments of existing fixed-point logics that have been developed for the classical BI solution concept. We also discuss changes in earlier modal logics expressing `surface reasoning' about best actions in the presence of short sight. Our analysis may point the way to logical and computational analysis of more realistic game models.

%Although this concept has been proposed in the study of game theory \cite{GrossiT12:Sight}, many critical issues concerning the features of short-sight and resulting changes are left open.

%our model and results for single-agent decision-making can be naturally extended to multi-player games. some logical systems for classical games are extended to fit into games with short sight.
\end{abstract}

%%%%%%%%%%%%%%%%%%%%%%%%%%%%%%%Introduction%%%%%%%%%%%%%%%%%%%%%%%%%%%%%%%%%%%%%%%%%%%%%%%%%%%%%%%%%%%%%%%%%%%%%%%%%%%%%%%

\section{Introduction}\label{section:introduction}
There is a growing interest in the logical foundations, computational implementations, and practical applications of \emph{single-agent sequential decision-making} (SSDM) problems \cite{Peterson2009,Houlding08,North68atutorial,Hansson94decisiontheory,BonetBretto2004,Littman1996,Kjetil2001} in such diverse areas as Artificial Intelligence, Control, Logic, Economics, Mathematics, Politics, Psychology, Philosophy, and Medicine.
%This development can be credited to a pressing need in above deciplines to formulate a best strategy comprising a series of decisions which involve every resulting state of the world.
Making decisions is central to agents' routine and usually, they need to make multiple decisions over time. Indeed, a current situation is a result of past sequentially linked decisions,  each impacted by the preceding choices.
%For modeling and analyzing such a sequential decision-making process, various formalisms like decision trees, influence diagrams or Markov decision processes \cite{BonetBretto2004,Littman1996,Kjetil2001} can be employed.

It is quite natural in sequential decision-making scenarios, particularly, in large systems, that agents may have some uncertainties and limitations on their precise view of the environment. The current literature   \cite{BonetBretto2004,North68atutorial} has studied uncertainty which an agent faces in recognizing possible outcomes after taking an action and the probabilities associated with these outcomes, as well as the partial observability of what the actual state is like. In addition to these, a  realistic aspect that affects a SSDM process is the short-sightedness of the agent, which blocks a full view of all the available actions. Short sight plays a critical role in such a situation, since, while making a  choice, the ability to foresee a variety of alternatives  and predict future decision sequences for each  of them, may make a significant difference. Nonetheless, such restrictions have not been discussed systematically yet in decision theory or game theory.

In \cite{GrossiT12:Sight},  a game-theoretic framework called \emph{games with short sight} was proposed. This framework explicitly models players's limited foresight in extensive games and calls for a new  solution termed as \emph{Sight-Compatible Backward Induction} (SCBI).  However, many essential issues related to sight remain unclear, such as: What is the exact role of sight? Will the outcome be better when sight is larger? What is the relation between  SCBI and classical \emph{backward induction}(BI)?  There are also unexplored issues pertaining to logical aspects. Which minimal logic is needed for formally characterizing a short-sight framework? Are existing logics for BI still applicable, or can they be extended to fit short-sight scenarios?  How different are the logical properties of the game frames for SCBI and for BI? Without such a logical analysis, the framework of \cite{GrossiT12:Sight} does not suffice for disclosing the general features of short sight and the changes it brings about in thinking about decisions and games. Additionally, in  multi-player games, short sight has to interact with many other factors, such as agents' mutual knowledge and interactive decisions and moves.

Having said this, we still start by focusing on short sight in single-agent sequential decision-making process. For this, we propose a model of `\emph{preference-sight trees}' (P-S trees). As the term says, a P-S tree combines the agent's \emph{ preference} and  its sight, as both are essential to decision problems  \cite{2011Rossi}.  We will study how the two are correlated, and cooperate to act on decision-making processes and their final outcomes.

As a preliminary illustration, consider the connection between larger sight and better outcome. A first impression might be that an agent will always perform better with larger sight. Surprisingly, this is not always true.   Sometimes, one can see much further into the future but receive a small payoff, while having one's vision restricted to a limited set of future alternatives yields a better payoff.

\vspace{-1ex}

\begin{example}\label{seemoreseeless} $Alice$ has to make sequential decisions at two stages (shown in Figure \ref{Figure:seemoreseeless}). For each stage, she can choose either $L$ or $R$.  Assume that the preference order (from most preferable to least preferable) among the four outcomes is $RR, LL, RL, LR$.  Now consider  two cases:

Case 1. At the start, $Alice$ sees two paths, $LR$ and $RL$. She chooses $R$ since it initiates $RL$ which is preferable to $LR$. At the second-stage, $Alice$ then foresees $RR$ and $RL$. She happily makes the best decision $RR$.

Case 2. $Alice$ sees more, e.g., $LL$, $LR$, and $RL$, immediately at the first stage. Therefore she thinks that $L$ is a better initial choice than $R$. Consequently, at the second-stage, she can only choose from $LL$ and $LR$.

Conclusion: Even though $Alice$ could see more in Case 2, she ultimately obtains a less preferable outcome.
\end{example}

\vspace{-3ex}

\begin{figure}[h]
  \begin{center}
  \includegraphics[width=120pt]{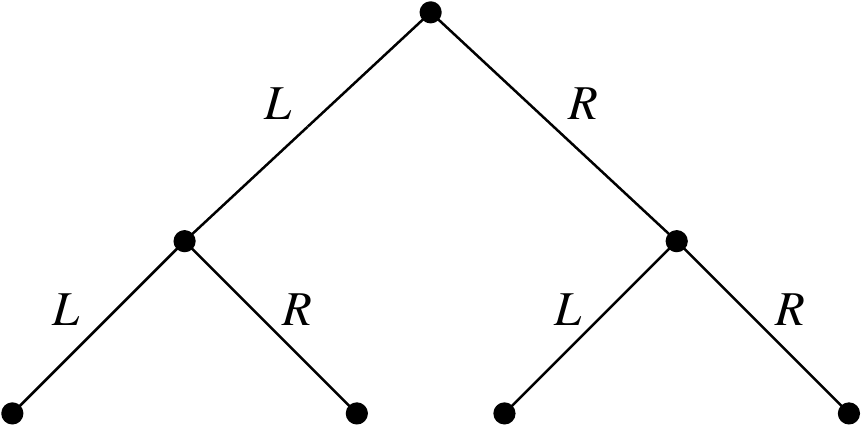}
  {\caption{Two-stage decision-making\vspace{-4ex}}\label{Figure:seemoreseeless}}
  \end{center}
\end{figure}

\vspace{-3ex}

%The fundamental explanation of the above example comes from the performance of preference and sight. More specifically, preference undergoes change because of short sight.  At the first stage, $Alice$'s sight causes her to prefer $R$ to $L$, since $R$ initiates $RL$ which is preferable to $LR$. This preference is consistent with the whole picture, since $RR$ is the best choice. However, $Bob$ can see $LL$ which makes him think that $L$ is a better initial choice than $R$. Consequently, he misses the best choice.

This example demonstrates some of the crucial features that govern SSDM situations:

1) What an agent can foresee plays a crucial role in the decision-making process, since her sight determines the set of available choices.

2) Sight also updates her preferences over the options, and thereby the outcomes obtained in rational play.

3) Although in Case 2, \emph{Alice} does not get the best result, we can say that, given her sight, she plays optimally in a local sense. In other words, this is a rational plan for her, even though it is not equivalent to the rational outcome of classical decision theory or game theory \cite{Osborne2004}.

In this paper, we address all three challenges, but first we clarify our approach. To focus on sight, we ignore other factors such as the probability of moves by Nature. Also, we model the outcome of a decision as completely determined, or in other words, possible outcomes for each alternative and the probability corresponding to each outcome are encapsulated as a black box.

%%%%%%%%%%%%%%%%%%%%%%%%%%%%%%%%第二节%%%%%%%%%%%%%%%%%%%%%%%%%%%%%%%%%%%%%%%%%%%%%%%%%%%%%%%%%%%%%%%%%%%%%%%%%%%%%%
\section{Modeling Single-agent sequential decision-making}\label{section:preference-sight-tree}

We begin by defining a  structure called \emph{preference-sight tree} for modelling single-agent sequential decision-making (SSDM) processes.  Using this model, we then clarify the role that sight plays by discussing a series of changes it produces in agent's preferences,  decision-making procedures and their outcomes.

\subsection{Models}

There are two kinds of models for decision-making scenarios corresponding to two perspectives. One is an \emph{explicit} model from the perspective of Nature, or an outsider/designer; the other is the \emph{implicit} model from the perspective of the agent involved, or an insider/decider.
The former is complete and perfect in the sense that the outsider  holds a full view of all the options together with the \emph{objective} quality of these options, and thus can explicitly specify the reward of each situation for the decision-maker. In contrast with this,  the latter's views are possibly limited to a near future, especially in large-scale surroundings. Moreover, owing to limited foresight, the agent may also reason mistakenly about the quality of different choices, leading to what we call \emph{subjective} preference.

Both the above perspectives are essential:  the former offers a whole picture of the environment, the latter shows the actual play of the decider. In this section, we first introduce an explicit model of \emph{preference trees}.  After this, by endowing such trees with the agent's view of the process and his/her subjective preference in this view, we formulate an integrated model of \emph{preference-sight trees} which allows us to model both perspectives together.

\subsubsection{Preference trees (P trees)}

A preference tree is a decision tree with only two elements: histories and preferences. Each history corresponds to a situation resulting from previous decision actions, and a preference represents the objective quality of each of these situations. To ensure the existence of backward induction solutions, we confine ourselves to finite histories.

\begin{definition}$($Preference tree$)$\label{preference-tree} A  {\bf preference tree} is a tuple $T= (H,\succeq)$ where

\begin{itemize}
\item $H$ is a non-empty set of finite sequences of actions, called \emph{histories}.

   $\circ$ The empty sequence $\varepsilon$ is a member of $H$;

   $\circ$ If $(a^k)_{k=1,...,K} \in H$ and $L<K$ then \makebox{$(a^k)_{k=1,...,L}\in H$};

\vspace{-1ex}

\item  $\succeq$ is a total order over $H$. %for any $l\leq l(H)$.
\end{itemize}
\end{definition}

\vspace{-1ex}

Let $A$ denote the set of all actions.  Any history $h$ can be written as a sequence of actions: ($a^k)_{k=1,\ldots, n}$, where each $a^k\in A$. If there is no $a^{n+1}$ s.t. $(a^k)_{k=1,...,n+1}\in H$, then history ($a^k)_{k=1,\ldots ,n}$ is  a terminal one. The set of terminal histories is denoted $Z$. The set of actions that are available at $h$ is denoted  $A(h)\subseteq A$. For any histories $h,h'$, if $h$ is a prefix of $h'$ we write $h \lhd h'$. The strict part of $\succeq$ is $\succ$, with  $h_1\succ h_2$ if $h_1\succeq h_2$ and not $h_2\succeq h_1$ for any two histories $h_1$ and $h_2$. Accordingly, $h_1\sim h_2$ iff $h_1\succeq h_2$ and $h_2\succeq h_1$.

\vspace{1ex}

Several remarks need to be made on the role of preference relations in the above definition:

(1) Instead of defining preference merely over terminal histories, we have defined it over all histories, an idea going back to  \cite{Harrenstein03}. Here preference over intermediate histories is necessary for our aim of modelling an agent's decision-making under limited foresight, which usually consists of intermediate histories.

%(2) We believe that every non-terminal state has a payoff which represents the value of its local situation. This value may come  from all kinds of criteria or goals, \cite{Liu11}, such as money spent and earned in the course of action, available resources, or heuristic prospects of winning or losing. Thus, the significance of preference over intermediate histories may come from explicitly representing criteria that also apply to intermediate stages. In particular, we make no strong assumptions on how this preference over intermediate states is related to classical preference over terminal states.

(3) For convenience, we do not strictly differentiate the two main views of preference: qualitative and quantitative. Although we use qualitative order generally, we sometimes switch to numerical payoff when it is advantageous.\footnote{There is a debate on whether preference and utilities are the same \cite{Hansson94decisiontheory,BERMUDEZ1994}. Here we adopt the operational understanding of utility and do not distinguish it from preference.}

\subsubsection{Preference-sight Trees (P-S trees)}

P tree is an explicit model for decision-making scenarios which is independent of an agent. However, for an agent, the tree may appear differently in his/her limited view.  \cite{GrossiT12:Sight} proposes the idea of short sight, where the authors use a sight function to denote the set of states that players can actually see at every position in an extensive game.  Let us start by adapting their technique to preference trees.

\vspace{-1ex}

\begin{definition} \label{sight-function} Let $T=(H,\succeq)$ be a  preference tree.  A {\bf sight function} for $T$ is a function $s: H\rightarrow 2^H\backslash\{\emptyset\}$ satisfying $s(h) \subseteq H|_h$ and $|s(h)|< \omega$, where $H|_h$ represents the set of histories  extending $h$. As a special case, $h\in H|_h$.
\end{definition}

\vspace{-2ex}

In words, the function $s$ assigns to each history $h$ a finite subset of all available histories extending $h$.

\vspace{1ex}

The first effect that sight produces is as follows: Given a P tree, for any history $h$, a sight function always gives us a restricted tree.

\begin{definition} Let $T=(H,\succeq)$ be a  preference tree.  Given any history $h$ of $T$, a \textbf{visible tree} ${T}_h$ of $T$ at $h$ is a tuple $(H_h, {\succeq}_h)$, where $H_h=s(h)$, i.e., $H_h$ captures the decider's view of the decision tree; $\succeq_h$ represents  the subjective preference over $H_h$.
\end{definition}

A visible tree is actually an implicit model in our earlier terms.  $H_h$ also contains a set of terminal histories $Z_h$, which are those without successors in $s(h)$. Note that typically, the $Z_h$ are non-terminal for $T$.

\vspace{1ex}

Further, the preference order ${\succeq}_h$ is different from the objective preference $\succeq$. In fact, the formation of ${\succeq}_h$ is an update via a bottom-to-top process in terms of an agent's sight. This updating process involves leaving the payoffs of $Z_h$ as the same as their objective payoffs, then updating the payoffs of other histories in $H_h$ backwards, starting from the leaf nodes and proceeding towards the root of the tree.

\vspace{1ex}

The reason why we employ such an updating process is that, while the objective payoffs reflect the goodness of these situations, they are not the actual reward that an agent can get if he/she chooses this option.  At each decision point,  the subjective payoff of one available option is inherited from the best reachable terminal histories of the current visible tree. Therefore, the preference relation $\succeq_h$ in $T_h$ is not always consistent with the preference relation $\succeq$ in $T$.

\vspace{1ex}

This updating process is described by Algorithm 1:

\vspace{1ex}

{\small \textbf{*}For convenience, here we use payoffs $P$ to represent rewards.}
%\vspace{-1ex}
\begin{algorithm}
{
\caption{\small{\textbf{Preference updating in visible trees}}}\label{preferenceupdating} PU($T,h,s$)\\
\KwIn{ A  P tree $T=(H,\succeq)$ (or $T=(H,P)$), current history $h$, and a sight function $s$}
\KwOut{ A visible tree $T_h=(H_h, \succeq_h)$ or ($T_h=(H_h, P_h)$)}
\Begin{
   $H\cap s(h)\rightarrow H_h$;\\
   \For {any $z\in Z_h$ \tcc*[f]{Keep the payoffs of terminal histories unchanged}}{
   $P(z)\rightarrow P_h(z)$; $~~1\rightarrow \textit{flag}[z]$;\\
   }

   \While {$\textit{flag}[h]==0$}{
     \For {any $h' \in H_h$}{
       \If {(for all $(h'a)\in H_h$, $\textit{flag}[(h'a)]==1$) \tcc*[f]{If all of its children have been visited, reset its payoff as the highest one among them}\\}
       {$\textit{max}\{P_h(h'a)\}\rightarrow P_h(h')$; $~~1\rightarrow \textit{flag}[h']$;\\}
     }
    }
   Return $T_h$;
   }
}
\end{algorithm}

%\vspace{-1ex}

%\begin{example} In Figure \ref{Figure:preferenceandsight}, the left is a P tree $T$ and the number beside each state represents the objective payoff of that history. The right is the visible tree $T_\varepsilon$ at initial history $\varepsilon$. In $T$, the objective payoffs of $h_1$ and $h_2$ are 3 and 2 respectively, so $h_1\succeq h_2$. However, things are changed due to short-sight: In $T_\varepsilon$, the subjective payoffs of $h_1$ and $h_2$ are $1$ and $2$ respectively, and thus $h_1$ ${\preceq}_\varepsilon$ $h_2$.
%\end{example}

%\vspace{-3ex}

%\begin{figure}[h]
 % \begin{center}
 % \includegraphics[width=220pt]{preferenceandsight}
 % {\vspace{-1ex}\caption{Preference in visible tree}\vspace{-4ex}\label{Figure:preferenceandsight}}
 % \end{center}
%\end{figure}

%\vspace{1ex}

%In particular, each visible tree is a preference tree:
\begin{fact} Let $T=(H, \succeq)$ be a P tree. Each visible tree ${T}_h=(H_h, {\succeq}_h)$ is a P tree.
\end{fact}

%\begin{proof} We need to show that the visible tree  ${T}_h=(H_h, {\succeq}_h)$ at $h$  is a preference tree.
%This can be done by showing that $H _h$ and ${\succeq}_h$ fit Definition \ref{preference-tree}:

%(1) $H_h$ is a nonempty set of finite histories. Here, the non-emptiness is simply ensured by the fact that $s(h)$ is always nonempty, while finiteness follows from the finiteness of histories in $H$.

%(2) By Algorithm \ref{preferenceupdating}, ${\succeq}_h$ is a total order over $H_h$. $\Box$
%\end{proof}

%\vspace{-0.5ex}

Correspondingly, we denote the prefix relation in $T_h$ by $\lhd_h$, and the actions that are available at $h$ by $A_h(h)$.

\vspace{0.5ex}

Finally we proceed to define our model of preference-sight trees.  A preference-sight tree allows us not only to represent the outsider's view, i.e., $(H,\succeq)$, but also to derive a series of implicit models, i.e., $(H_h, \succeq_h)$,  one for each $h$.

\begin{definition}$($Preference-sight tree$)$ A \emph{\textbf{preference-sight tree}} (P-S tree) is a tuple $(T,s)$, where $T=(H,\succeq)$ is a  preference tree and $s$ a sight function for $T$.
\end{definition}

In P-S trees,  an agent's sight should satisfy the following properties: First, if an agent can see a given future history, then he/she can also see any intermediate history up to that point.  Second, if the agent can see a history two steps forward, then after moving one step ahead, he/she can still see it. These features are formally stated as follows.

\vspace{0.5ex}

\begin{fact} $($Properties of sight function$)$ Let $(T,s)$ be a  P-S tree. For all $h,h', h''\in H$, with $h\lhd h'\lhd h''$, $s$ satisfies :
\begin{itemize}
\item[$\emph{DC}$]\textit{(Downward-Closed)}:  if $h'' \in s(h)$, then $h'\in s(h)$.
\item[$\emph{NF}$]\textit{(Non-Forgetting)}:   if $h''\in s(h)$, then $h''\in s(h')$.
\end{itemize}
\end{fact}

\subsection{Solution concepts}

Solution concepts are at the center of all choice problems.  In what follows, we define two solution concepts for P-S trees,   adapted from \cite{Osborne1994,GrossiT12:Sight}.   After this, we investigate the conditions for their equivalence.

\subsubsection{BI history and SCBI history}

Backward Induction (BI) is well-known in game theory \cite{Osborne1994}. The process runs like this.  First, one determines the optimal strategy of the player who makes the last move of the game. Using this information, one can then determine the optimal action of the next-to-last moving player. The process continues backwards in this way  until all players' actions have been determined in the whole game.  Its adaptation to single-agent decision-making process becomes a maximality problem for the agent involved.

In a P-S tree, we say that one history $h$ is $\emph{max}_{\succeq}$ in a set of histories $\Gamma\subseteq H$, if $h\in \Gamma$ and for any other history $h'$ in $\Gamma$, it holds that $h\succeq h'$, and we write this as $h\in \emph{max}_{\succeq} \Gamma$. The strict part for $\emph{max}_{\succeq}$ is $\emph{max}_{\succ}$.

\vspace{-1ex}
\begin{definition}$($\emph{BI} history$)$ Let $(T,s)$ be a P-S tree. A history $h^*\in Z$ is a \emph{\textbf{BI history}} of $T$, iff  $h^*\in \emph{max}_{\succeq}Z$. Also, we use ${\textbf{\emph{BI}}}$ to denote the set of BI histories in $T$.
\end{definition}

\vspace{-1ex}

A BI history of a P-S tree is a terminal history that is most preferable or equivalently, that has a maximal payoff.

Backward induction precludes short-sight, while in practice it is impossible for an agent to foresee all final outcomes all the time.  In \cite{GrossiT12:Sight},
a new solution concept was proposed to capture optimal play of short-sighted players: \emph{sight-compatible subgame perfect equilibrium}. The main idea is that at each decision point, the current player chooses a locally optimal move by a local BI analysis within the visible part.
Here, we adapt this notion to P-S trees, yielding the \emph{sight-compatible backward induction history}.

\begin{definition}$($\emph{SCBI} history$)$\label{SCBI} Let $(T,s)$ be a P-S tree.  A history $h^*\in Z$ is a \emph{\textbf{Sight-Compatible Backward Induction history}} $($\emph{SCBI history}$)$ of $T$, iff for each history $h$ with $h\lhd h^*$, and the action $a$ following $h$, i.e., $(ha)\lhd h^*$, we have that $\exists z\in \emph{max}_{\succeq}Z_h$ such that $(ha)\lhd z$. Also, we use ${\textbf{\emph{SCBI}}}$ to denote the set of SCBI histories in $T$.
\end{definition}

\vspace{-1ex}

The difference  between SCBI and BI histories is obvious. A BI history is one with highest payoff among the set of terminal histories in the P-S tree, while for a SCBI history every restriction of it should be a local BI history for the visible tree.  Thus,  BI histories are the BI outcomes for the objective model {$(H, \succeq)$}, while SCBI histories are a combination of best responses to all subjective models $(H_h, \succeq_h)$. Typically it is the case that $\textbf{SCBI}\neq \textbf{BI}$.

\vspace{-1ex}

\begin{example} Consider the P-S tree $(T,s)$ in Figure \ref{Figure:counter-example-2}, where $s(\varepsilon)=\{L\}$, and $s(L)=\{LR\}$. It is easy to check that $\textbf{\emph{BI}}\neq \textbf{\emph{SCBI}}$, since $\textbf{\emph{BI}}=\{LL\}$, while $\textbf{\emph{SCBI}}=\{LR\}$.
\end{example}

\vspace{-2ex}

\begin{figure}[h]
\vspace{-2ex}
  \begin{center}
  \includegraphics[width=100pt]{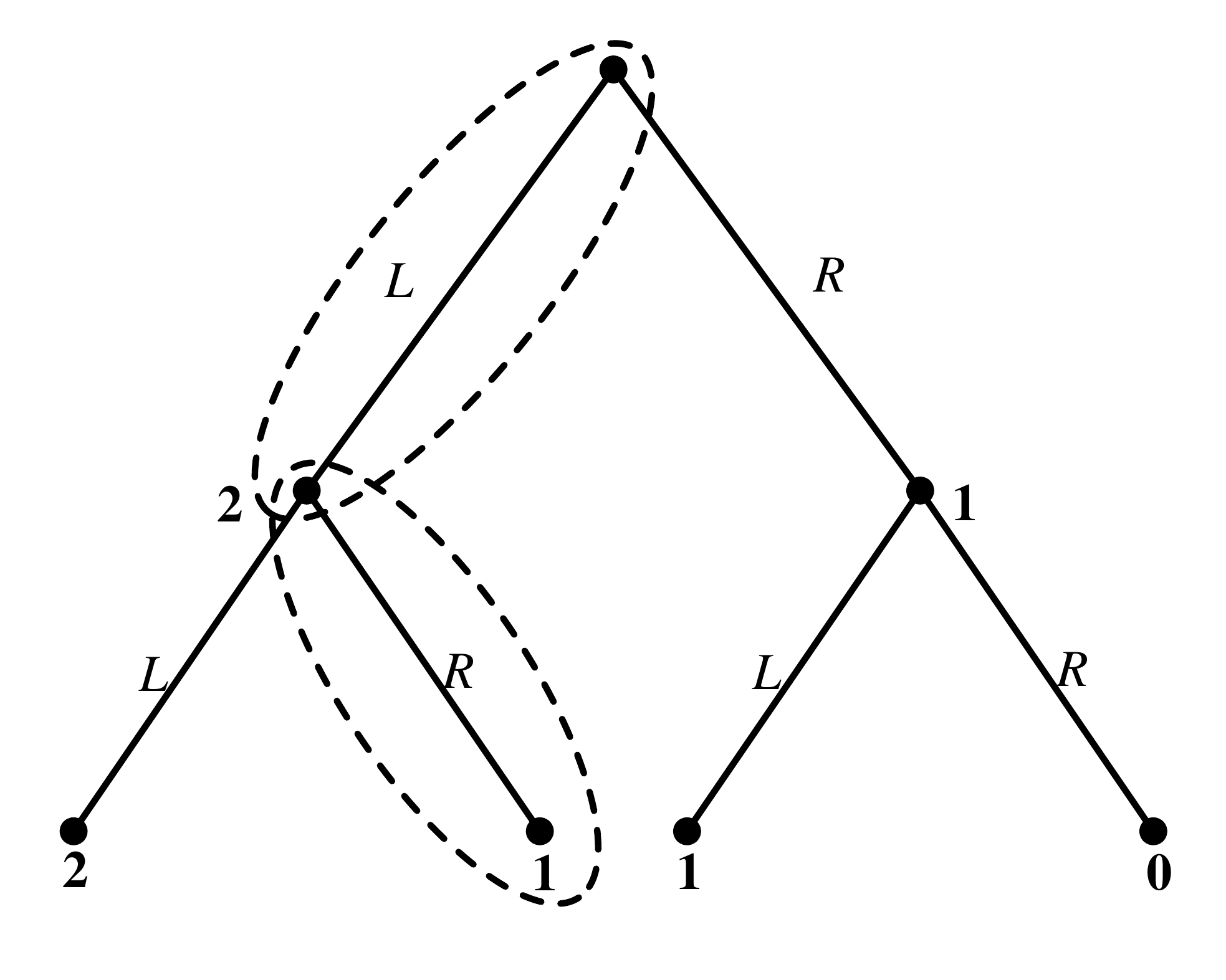}
  {\vspace{-3ex}\caption{${\textbf{BI}} \neq {\textbf{SCBI}}$ }\label{Figure:counter-example-2}}
  \end{center}
\end{figure}

\vspace{-3ex}
However, sometimes the two notions can be  equivalent.

\vspace{-1ex}

\begin{example} Consider a P-S tree, with $T$ and $s$ shown by Figure \ref{Figure:counter-example-1} (a), and  Figure \ref{Figure:counter-example-1} (b) respectively.  In (b) the three dotted circles represent $s(\varepsilon)$, $s(L)$ and $s(R)$. For histories $L$ and $R$, their objective payoffs in $(a)$ are $1$ and $2$, respectively. However, in $T_\varepsilon$, the subjective payoff of $L$ is updated to $3$ and $R$ to $2$. Obviously,  $\textbf{\emph{BI}} = \textbf{\emph{SCBI}}=\{LL\}$.
\end{example}

\vspace{-4ex}

\begin{figure}[h]
  \begin{center}
  \includegraphics[width=230pt]{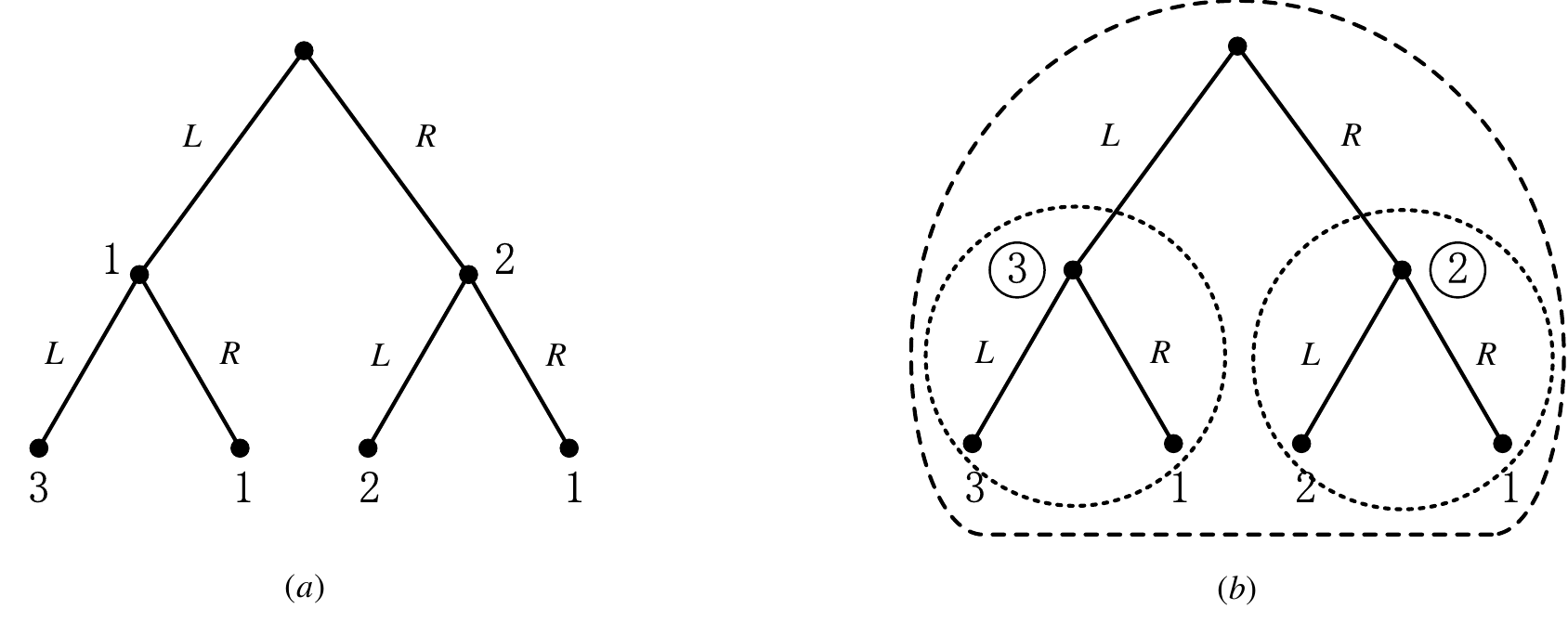}
  {\vspace{-3ex}\caption{${(\textbf{BI}}$ = ${\textbf{SCBI}})$\vspace{-16ex}}\label{Figure:counter-example-1}}
  \end{center}\vspace{-3ex}
\end{figure}

\subsubsection{Equivalence condition}

Then an interesting question on BI and SCBI histories arises: are there  conditions under which the two will be equivalent?  To get a feeling for this, a first attempt at an answer looks for a condition related to consistency between subjective and objective preferences.

Two histories are said to be `preference-sight consistent' if the subjective preference in each sight-restricted tree is consistent with the objective preference over them:

\begin{definition}\vspace{-1ex}{\rm$($Preference-sight consistency$)$\label{PSconsistency} Let $(T,s)$ be a  P-S tree, and $T_h$ be the visible tree at an arbitrary history $h$. Then for any two  histories $h_1$, $h_2$  of $T_h$, we say $(h_1,h_2)$ satisfies {\bf{ preference-sight consistency}} at $h$ iff  }
\vspace{-0.5ex}
 $$h_1\succeq h_2 \emph{ iff } h_1  ~{\succeq}_h ~h_2$$
\end{definition}

\vspace{-0.5ex}

{\rm If for any history $h\in T$, the pair of arbitrary two histories $(h_1, h_2)$ in $T_h$ is {\emph{preference-sight consistent}} (at $h$), then we say $(T,s)$ is {\emph{preference-sight consistent}}.}

\vspace{1ex}

Is preference-sight consistency an appropriate condition for ${\textbf{BI}}$ = ${\textbf{SCBI}}$? We have the following observation:
\begin{fact} Preference-sight consistency does not guarantee that ${\textbf{\emph{BI}}}$ = ${\textbf{\emph{SCBI}}}$.
\end{fact}

\begin{proof}Consider Figure \ref{Figure:counter-example-2}.  Suppose that $s(R)$ contains only one successor. Then it is easy to see that $(T,s)$  is preference-sight consistent. However, $\textbf{BI}\neq \textbf{SCBI}$.
\end{proof}

Next, does the other direction hold?
\vspace{-1ex}
\begin{fact} Preference-sight consistency does not follow from ${\textbf{\emph{BI}}}$ = ${\textbf{\emph{SCBI}}}$.
\end{fact}

\begin{proof}The situation in Figure \ref{Figure:counter-example-1} is a counterexample, in which  $\textbf{BI} = \textbf{SCBI}=\{LL\}$, but $(T,s)$ is not preference-sight consistent, since $R \succ L$ and $L~ {\succ}_{\varepsilon}~ R$.
\end{proof}

\vspace{-1ex}

What is the exact condition for ${\textbf{BI}}$ = ${\textbf{SCBI}}$? From the failure of preference-sight consistency, we can draw a lesson. In Figure \ref{Figure:counter-example-2}, the main reason for $(T,s)$ being inconsistent is that at history $L$, the branch $LL$, which in fact forms a BI history, is non-observable to the agent. This tells us that the one with maximal payoff should always be visible. Consider then the example in Figure \ref{Figure:counter-example-1}. Here all the options are within agent's sight, but we notice that although the path $LL$ following $L$ finally turns out to be better than that following $R$, which makes subjectively $L \succ_{\varepsilon} R$, the objective payoff of $L$ itself is lower than $R$. Thus, it fails to imply the consistency between preference and sight.

\vspace{1ex}

Based on the above analysis, we now isolate necessary and sufficient conditions for ${\textbf{BI}}$ = ${\textbf{SCBI}}$. First, we define an auxiliary property of \emph{sight-reachability}, which intuitively reflects whether each restriction of a history is visible.

\vspace{-1ex}

\begin{definition}$($Sight-reachability$)$ \label{sightr} A BI history $h^*$ is {\bf sight-reachable} if, for all $(ha)\lhd h^*$, we have  $(ha)\in H_h$, where $h,h'$ are histories, and $a$ is an action following $h$.
\end{definition}

\vspace{-2ex}

\begin{theorem}$($Equivalence Theorem$)$ For any P-S tree $(T,s)$, ${\textbf{\emph{SCBI}}}$= ${\textbf{\emph{BI}}}$ iff the following conditions are satisfied:

\begin{itemize}
\vspace{-2ex}
\item[\emph{I)}.] Any history $h^*\in \textbf{\emph{BI}}$ is sight-reachable.
\vspace{-1ex}
\item[\emph{II)}.] Any history $h^*\in \textbf{\emph{BI}}$ is locally optimal: For any history $(hh')\lhd h^*$, if $(hh')\in Z_h$, then  $(hh')\in \emph{max}_{\succ} Z_h$ and for any other $(hh'')\in Z_h$, $(hh')\sim (hh'')$ iff $\exists z\in \textbf{\emph{BI}}$ such that $(hh'')\lhd z$.

%\item[ii$)$.] All BI histories $h^*$ are locally optimal: For every BI history $h^*$, and all $(hh')\lhd h^*$, if $(hh')\in Z_h$, we have $(hh')$ is   $\emph{max}_{\succ}$ in ${Z_h\setminus}{\textbf{BI}}_h$, and $\emph{max}_{\sim}$ in ${\textbf{BI}}_h$.

\end{itemize}
\end{theorem}

\begin{proof}
\begin{itemize}
\item[($\Rightarrow$)] $\rm{I)}$. We show that every $h^*\in \textbf{BI}$ is sight reachable. That is, for all $(hh')\lhd h^*$, it holds that $(ha)\in H_h$.  By ${\textbf{SCBI}}$= ${\textbf{BI}}$, we know that any history $h^*$ in $\textbf{BI}$, is also in $\textbf{SCBI}$. By Definition \ref{SCBI}, for each of its prefix $h$, $h^*_h$ is $\emph{max}_{\succeq}$ in $Z_h$. So $h^*_h$ is in $Z_h$. In addition, by non-emptiness of $Z_h$, $h^*_h$ is not an empty sequence. Thus, for all $(ha)\lhd h^*$, it holds that $(ha)\in H_h$. So $h^*\in \textbf{BI}$ is sight-reachable.
    \vspace{-1ex}
\item[] To show condition $\rm{II)}$, take any $h^*$ in $\textbf{BI}$,  we  have that it is in $\textbf{SCBI}$. Thus, for all  $(hh')\lhd h^*$, if $(hh')\in Z_h$, then $(hh')$ is $\emph{max}_{\succeq}$ in $Z_h$. Moreover, for any $(hu)\in Z_h$ such that $(hh')\sim(hu)$, we have $(hu)$ is a prefix of a BI history, i.e., $(hu)\in \textbf{BI}_h$. For suppose not, then  $(hu)$ is not a prefix of SCBI history. Then it must be $(hh')\nsim (hu)$. Contradict.
    \vspace{-1ex}
\item[($\Leftarrow$)]  Suppose conditions $\rm{I)}$ and $\rm{II)}$ are satisfied. It suffices to show (a)``every BI history is SCBI history of $T$'', and (b) `` every SCBI history is BI history of $T$''.
    \vspace{-1ex}
\item[] For (a), take any BI history $h^*$. By $\rm{I)}$, all BI histories are sight reachable. Further by $\rm{II)}$, for all $(hh')\lhd h^*$, if $(hh')\in Z_h$, then $(hh')$ is $\emph{max}_{\succeq}$ in $Z_h$. This is to say that for each of its prefix $h$, $h^*_h$ is $\emph{max}_{\succeq}$ in $Z_h$.  By definition \ref{SCBI},  $h^*$ is a SCBI history.
\vspace{-1ex}
\item[] For (b), take any SCBI history $h^*$. We can show it is a BI history, i.e., $h^*$ is $\emph{max}_{\succeq}$ in $Z$. For suppose not, then there exists a BI history $h'$ such that $h'\succ h^*$. Notice that there must be some history $u$ which is the common prefix of $h^*$ and $h'$. Since $h'$ is a BI history, by condition $\rm{I)}$ and \rm{II)}, we know that $h'_{u} \succ h^*_{u}$. Then $h^*_{u}$ is not a prefix of a SCBI history. Thus, $h^*$ is not a SCBI history. Contradiction.
    \end{itemize}
    \end{proof}

%%%%%%%%%%%%%%%%%%%%%%%%%%%%%%第三节%%%%%%%%%on%%%%%%%%%%%%%%%%%%%%%%%%%%%%%%%%%%%%%%%%%%%%%%%%%%%%%%%%%%%%%%%%%%%%%%%

\subsubsection{More sight, better outcome?}

We have seen earlier on that,  SCBI may loss global optimality. The BI history definitely has a maximal payoff, while it might not be the case for SCBI, since each action is chosen with a limited sight. So $\textbf{BI}\succeq \textbf{SCBI}$ holds without exception, in the sense that any BI history is no worse than any SCBI history. One might conjecture that more sight always contributes to better outcomes. Yet, the fact below falsifies this.

\vspace{-1ex}

\begin{fact} Let $T$ be a P tree. Also, let $s_1$ and $s_2$ be two sight functions for $T$ satisfying $s_1(h)\subseteq s_2(h)$ for any history $h$ in $T$.  Take any two \emph{SCBI} histories $z_1$ and $z_2$ of $(T,s_1)$ and $(T,s_2)$ respectively. Then the following three cases are all possible: $a)$  $z_1\succ z_2$; \hspace{1ex} $b)$  $z_2\succ z_1$; \hspace{1ex} $c)$  $z_1\sim z_2$.

\end{fact}

\begin{proof} Case (a) has been shown in Example  \ref{seemoreseeless}. Case (b): Obviously, Figure \ref{Figure:counter-example-2} offers an instance for this. Case (c): The scenario depicted in Figure \ref{Figure:counter-example-1} is an  example.$\Box$
\end{proof}

\vspace{-1ex}

In conclusion, full sight guarantees a maximal payoff. However, with short sight, increase of sight does not always improve the outcome. The  added sight may bring misleading information, e.g., a branch which is temporarily nicer but actually unpromising, and finally gives rise to an even worse outcome. Still, this does not mean that SCBI is deficient: rather, these observations seem realistic for real agents. These issues will be  discussed further in Section 4.

\section{A Logical Analysis}
After modelling decision-making with short sight by preference tree models, it is instructive to see what a logical language looks like for reasoning about these models, especially the role of sight in a SSDM process.  So far, no such logic has been proposed, though logics of game-theoretic structures have been extensively studied -- see  \cite{Hoek06,Harrenstein02onmodal} --  while there are a few preliminary logic analyses of sight on its own, \cite{Degremont14,Chanjuan13}.  In this section, we design a minimal and natural logical system that supports reasoning about sight in the context of single-agent decision-making processes, characterizing basic properties of preference-sight trees, and formally capturing the results in the previous section.

\subsection{Syntax and Semantics}
To reason about the key ingredients (i.e., histories, preferences, and sights) of a P-S tree, we take $P^{(T,s)}$ as a set of propositional letters, which at least contains the following {\footnote{The idea of defining $\overline{h}$ is motivated by \cite{BaltagSZ09}, where the authors define an atomic sentence $\overline{o}$ for each leaf in a game tree.}:

\vspace{1ex}

$\bullet$ $\overline{h}$ for each history $h$.

\vspace{1ex}

$\bullet$ $\overline{h_1\geq h_2}$ encoding the preference relation of the agent over all histories, and the strict part of which is $\overline{h_1> h_2}$.

\vspace{1ex}

%$\bullet$ $\overline{h_1\geq_h h_2}$ corresponding to the preference relation in visible tree $T\lceil_h$, and similarly, $\overline{h_1>_h h_2}$ to the strict part.

$\bullet$ $\overline{s(h)}$  encoding the sight at each history $h$ in $T$.

\vspace{1ex}

Based on $P^{(T,s)}$, we give a language $\mathcal{L}$ for reasoning about P-S trees. In $\mathcal{L}$, we have a key dynamic operator $[!\varphi]$ for restricting to the worlds satisfying $\varphi$, and a universal modality with $A\varphi$ saying that $\varphi$ is true in every world.

\begin{definition}{\rm $($Preference-sight language$)$ Take any set of atomic letters $P^{(T,s)}$. The {\bf preference-sight language} $\mathcal{L}$ is given by the following BNF, where $p\in P^{(T,s)}$:}
\vspace{-0.5ex}
$$\varphi::= p~|\neg \varphi~|\varphi\wedge \psi~|[!\varphi]\psi ~|~A\varphi.$$
\end{definition}

We write $\langle!\varphi\rangle\varphi$ to abbreviate $\neg[!\varphi]\neg\varphi$.
\vspace{-0.5ex}

\begin{definition}{\rm $($Preference-sight models$)$ For  a P-S tree $(T,s)$, a \textbf{\emph{preference-sight model}} $M^{(T,s)}$ is a tuple $(H, \lhd, \mathcal{V})$ where the following holds:

\smallskip

$\bullet$ $H$ is the set of possible worlds, one for each history,
\smallskip

$\bullet$ $\lhd$ is the reachability (prefix) relation among worlds,
\smallskip

$\bullet$ $\mathcal{V}: P_T \rightarrow \rho(H)$ is an evaluation function satisfying:

\vspace{0.5ex}

\hspace{2ex}$(1)$ $\forall h\in H$, $\mathcal{V}(\overline{h})=\{h'|h'\lhd h\}$.

\vspace{-3.5ex}\begin{eqnarray*}
\hspace{-44ex}(2) ~\mathcal{V}(\overline{h_1\geq h_2}) =
\begin{cases}
H,       & \emph{IF}~ h_1\succeq h_2, \\
\emptyset, & Otherwise.
\end{cases}
\end{eqnarray*}

\vspace{-0.5ex}

\hspace{2ex}$(3)$ $\forall h\in H$, $\mathcal{V}(\overline{s(h)})= \bigcup\limits_{h'\in s(h)}\mathcal{V}(\overline{h'})$.}
\end{definition}

\vspace{-1.5ex}

Intuitively, $\overline{h}$ is true at all the worlds leading to $h$.  $\overline{h_1\geq h_2}$ is true everywhere if  $h_1\succeq h_2$, and nowhere otherwise.  Finally, $\mathcal{V}(\overline{s(h)})$ is a union of the worlds that make the given atom true for at least one element of $s(h)$.

\vspace{1ex}

There seems to be nothing striking in this syntax. However, given the special role of atoms, the natural model update differs from the usual one in dynamic-epistemic logic.

\vspace{-0.5ex}

\begin{definition}{\rm (Model update)\label{modelupdate} Given a preference-sight model $M^{(T,s)}=(H,\lhd,\mathcal{V})$ and a set $X\subset H$, the \textbf{updated model} $M^{(T,s)}_{!X}$ produced by the restriction of $X$ is defined as a tuple $(X,\lhd\cap X^2,\mathcal{V}_{!X})$, where \footnote{In this definition, $Z_X$ denotes the terminal histories in $X$, i.e., the set of histories that have no successors in $X$.}}
\end{definition}

\vspace{-4ex}

\begin{eqnarray*}
\hspace{1ex}\mathcal{V}_{!X}(p) =
\begin{cases}
\mathcal{V}_{!X} (\overline{h_1 \geq h_2}), & \emph{IF}~  p\emph{ is of the form} ~\overline{h_1\geq h_2} \\
%\{h| h\in X  ~\emph{and}~ h\lhd z_x\}, & \emph{IF}~  p\emph{ is in the form of}  ~\overline{z_x} ~ \emph{where} ~ z_x \in Z_X \\
%\bigcup_{h\lhd z_x, z_x\in Z_X}\mathcal{V}_{!X} (\overline{z_x}),   & \emph{IF}~  p\emph{ is in the form of} ~\overline{h} ~\emph{where}~ h\notin Z_X \\
\mathcal{V}(p)\cap X, & Otherwise
\end{cases}
\end{eqnarray*}

\vspace{-2ex}

\begin{eqnarray*}
\hspace{-2ex}\mathcal{V}_{!X}({\overline{h_1 \geq h_2}}) =
\begin{cases}
X,       & \emph{IF}~ {\mathcal{V}(\overline{z_1\geq z_2}})=H, ~\emph{where}\\
%&~\emph{where}~ z_1\in \emph{max}_{\succeq} Z_1, z_2\in \emph{max}_{\succeq} Z_2,\\
&~ z_1\in \emph{max}_{\succeq}\{z\in Z_X| h_1\lhd z\}, \\
&~ z_2\in \emph{max}_{\succeq}\{z\in Z_X| h_2\lhd z\} \\
\emptyset, & Otherwise
\end{cases}
\end{eqnarray*}

$M^{(T,s)}_{!X}$ is the update of the model  $M^{(T,s)}$  restricting the set of states to $X$, and the valuation function accordingly.  But crucially, the valuation for preference atoms in the new model reflects the updating process in the visible tree of Algorithm \ref{preferenceupdating}. In the following, we omit superscripts ${(T,s)}$.
%1) since the terminal histories in restricted model are updated, it leads to a redefinition of the evaluation of terminal-history atoms. And thus, the evaluation of non-terminal history atoms are also redefined. 2) the evaluation :

\smallskip

The semantics for this language is basically standard, \cite{Blackburn2001b}, so we only mention the truth condition of $[!\varphi]\psi$:

\smallskip

Let $M$ be a preference-sight model. For any state $h$ in $M$,

\vspace{-2ex}

$$M, h\models [!\varphi]\psi ~\emph{iff} ~ M, h \models \varphi \Rightarrow M_{!\varphi}, h\models \psi.$$

Validity of formulas is defined as usual, cf. \cite{Blackburn2001b}.

\subsection{Main characterization results}

Despite its simplicity, $\mathcal{L}$ can express our results in previous sections concerning  properties and solutions of P-S trees.  We introduce some helpful syntactic abbreviations, and then state our main characterization results.

\medskip

$\bullet$ $\overline{Z_h}=\bigvee\{~\overline{z}~|~z\in Z_h\}$.

\smallskip

$\bullet$ $\overline{\emph{max}_{\geq}X}$=$\bigvee\{~\overline{h}~|~h\in X, ~\emph{and}~ h\succeq h' ~\emph{for}~ \forall h'\in X\}$.

\smallskip

$\bullet$ $\overline{\rm{BI}}=\bigvee\{~\overline{z}~|~z\in \textbf{BI}\}$ ($\overline{\rm{BI}}$ holds at $T$'s BI histories).

\smallskip

$\bullet$ $\overline{\rm{SCBI}}=\bigvee\{~\overline{z}~|~z\in \textbf{SCBI}\}$, that is, the formula $\overline{\rm{SCBI}}$ holds at  the SCBI histories of $T$.

\begin{proposition} Let $(T,s)$ be a P-S tree and $M$ be a $\mathcal{L}$-model for it. Then $(T,s)$ is \emph{preference-sight consistent} iff the following formula is valid in $M$:

\vspace{-1ex}

$$\bigwedge\limits_{h}\bigwedge\limits_{h_1\in H_h}\bigwedge\limits_{h_2\in H_h}((\overline{h_1\geq h_2} \rightarrow [!\overline{s(h)}]\overline{h_1\geq h_2})\wedge$$

\vspace{-2ex}

$$(\langle!\overline{s(h)}\rangle\overline{h_1\geq h_2}\rightarrow \overline{h_1\geq h_2}))$$

\end{proposition}

%\begin{proof} Both directions $(\Rightarrow)$ and $(\Leftarrow)$  follow immediately from Definition \ref{PSconsistency} and the semantics of the language $L$.
%\end{proof}

\begin{lemma} For any P-S tree $(T,s)$ and model $M$ for it, a $BI$ history $h^*$ is \emph{sight-reachable} if and only if the following formula holds in $M$:

\vspace{1ex}

%$M^T\models$ $\bigwedge_{(ha)\lhd h^*, a\in A_h}(h\rightarrow [!s(h)]\overline{(ha)})$.

$(SR):~~$ $\bigwedge\limits_{h}\bigwedge\limits_{a\in A(h)}(A(\overline{(ha)}\rightarrow\overline{h^*})\rightarrow (A(\overline{(ha)}\rightarrow \overline{s(h)})))$.

\end{lemma}

\begin{proof} $(\Rightarrow)$ Suppose that $\rm{BI}$ history $h^*$ is sight-reachable. By Definition \ref{sightr}, we have that, for all $(ha)\lhd h^*$, it holds that $(ha)\in s(h)$, where $h,h'$ are histories, and $a$ is an action following $h$.  More formally, $(ha)\lhd h^*$ can be defined by the formula $A(\overline{(ha)}\rightarrow\overline{h^*})$ in the sense that, in $T$, for all $h$ and $a\in A(h)$, $(ha)\lhd h^*$ iff $M\models A(\overline{(ha)}\rightarrow\overline{h^*})$. And similarly $(ha)\in s(h)$ is defined by $A(\overline{(ha)}\rightarrow \overline{s(h)})$. Thus if a $\rm{BI}$ history $h^*$ is sight-reachable, then $M\models$ $\bigwedge_{h}\bigwedge_{a\in A(h)}(A(\overline{(ha)}\rightarrow\overline{h^*})\rightarrow (A(\overline{(ha)}\rightarrow \overline{s(h)})))$. The other direction can be proved in a similar way. $\Box$
\end{proof}

\begin{lemma} Let $(T,s)$ be a P-S tree and $M$ be a $\mathcal{L}$-model for it. A $\emph{BI}$ history $h^*$ is \emph{locally optimal} iff the following formula is valid in $M$:

\vspace{1ex}

$(LO):~~$\vspace{-4ex}$$(\bigwedge\limits_{h}\bigwedge\limits_{(hh')\in Z_h}(A(\overline{(hh')}\rightarrow \overline{h^*})\rightarrow$$
\vspace{-1ex}
$$(A(\overline{(hh')}\rightarrow \overline{\emph{max}_{\succeq}Z_h}) \wedge$$
\vspace{-2ex}
$$\bigwedge\limits_{(hh'')\in Z_h} (\overline{(hh')\sim (hh'')}\leftrightarrow \bigvee\limits_{z\in \textbf{BI}}(A(\overline{(hh'')}\rightarrow \overline{z}))))).$$

%$$((\overline{(hh')}\rightarrow \overline{\emph{max}_{\succeq}(Z_h\backslash \emph{BI}_h))}\wedge (\overline{(hh')}\rightarrow \overline{max_\sim \emph{BI}_h)}))).$$

\end{lemma}

%\begin{proof} ($\Leftarrow$) Suppose $BI$ history $h^*$ is locally optimal. Then for $(hh')\lhd h^*$, if $(hh')\in Z_h$, we have $(hh')$ is   $\emph{max}_{\succ}$ in ${Z_h\setminus}{\textbf{BI}}_h$, and $\emph{max}_{\sim}$ in ${\textbf{BI}}_h$. Similar with the above proposition, $A(\overline{(hh')}\rightarrow \overline{h^*})$ captures that $(hh')\lhd h^*$. While $((\overline{(hh')}\rightarrow \overline{\emph{max}_>(Z_h\backslash \emph{BI}_h))}\wedge (\overline{(hh')}\rightarrow \overline{max_\sim \emph{BI}_h)})))$ demonstrates that $(hh')$ is   $\emph{max}_{\succ}$ in ${Z_h\setminus}{\textbf{BI}}_h$, and $\emph{max}_{\sim}$ in ${\textbf{BI}}_h$.  Direction ($\Rightarrow$) is a similar check.
%\end{proof}

\vspace{0.5ex}

\begin{proof} ($\Leftarrow$) Suppose $\rm{BI}$ history $h^*$ is locally optimal. Then for $(hh')\lhd h^*$, if $(hh')\in Z_h$, we have $(hh')$ is   $\emph{max}_{\succeq}$ in $Z_h$. And for any $(hh'')$, $(hh'')\sim (hh')$ iff $\exists z\in \textbf{BI}$ s.t. $(hh'')\lhd z$. Similar with the above proposition, $A(\overline{(hh')}\rightarrow \overline{h^*})$ captures that $(hh')\lhd h^*$. And $A(\overline{(hh'')}\rightarrow \overline{z})$ shows that $(hh'')\lhd z$. Finally, $(\overline{(hh')}\rightarrow \overline{\emph{max}_{\succeq}Z_h})$ demonstrates that $(hh')$ is   $\emph{max}_{\succeq}$ in ${Z_h}$.  Direction ($\Rightarrow$) uses a similar check.
\end{proof}

%$\{h^*\in BI| ~\emph{if}~ \emph{for} ~\emph{all} ~ (ha)\lhd h^*,  (ha)\in H_h\}$.

%The following proposition is a logical characterization of the Equivalence-Theorem.

\begin{proposition}$($$\mathcal{L}$-characterization of equivalence$)$ Let $(T,s)$ be a preference-sight tree and $M$  a model for it. Then the following formula is valid in $M$:

\vspace{-2ex}

$$\models (A(\overline{\emph{BI}}\leftrightarrow \overline{\emph{SCBI}}))\leftrightarrow$$
\vspace{-3ex}
$$\bigwedge_{h^*\in Z}((A(\overline{h^*}\rightarrow \overline{\emph{BI}}))\rightarrow (SR\wedge LO))$$

\end{proposition}

\begin{proof}

\begin{itemize}

%\vspace{2ex}

\item[] \hspace{-6ex}Direction ($\Rightarrow$). We need to prove the following:
\item[] 1) $(A(\overline{\rm{BI}}\leftrightarrow \overline{\rm{SCBI}}))\rightarrow$$\bigwedge_{h^*\in Z}(A(\overline{h^*}\rightarrow \overline{\rm{BI}})\rightarrow SR)$.

\item[] 2) $(A(\overline{\rm{BI}}\leftrightarrow \overline{\rm{SCBI}}))\rightarrow$$\bigwedge_{h^*\in Z}(A(\overline{h^*}\rightarrow \overline{\rm{BI}})\rightarrow LO)$.
\end{itemize}

\smallskip

\noindent For 1). It is equivalent to prove that, for any $h^*\in Z$, $(\overline{\rm{BI}}\leftrightarrow \overline{\rm{SCBI}}) \wedge (A(\overline{h^*}\rightarrow \overline{\rm{BI}}))\rightarrow SR$. Suppose $\neg(SR)$. Then $\exists(ha)\lhd h^*$, and $(ha)\notin T_h$, and so, at $h$, the branch leading to $h^*$ is not visible in $T_h$.  Thus, the BI history in $T_h$ could not be a branch leading to $h^*$. By the definition \textbf{SCBI}, it follows that $h^*\notin \textbf{SCBI}$.  However, by $\overline{h^*}\rightarrow \overline{\rm{BI}}$ we know that $h^*$ is a BI history. This contradicts $\overline{\rm{BI}}\leftrightarrow \overline{\rm{SCBI}}$.

\smallskip

%\item[] For 2). Suppose not. Then $\exists h^*$, and $\exists (hh')\lhd h^*$, s.t. a): $(\overline{h^*}\rightarrow \overline{BI})$, and b): $not$ $(\overline{(hh')}\rightarrow \overline{\emph{max}_>Z_h\backslash \emph{BI}_h})\wedge (\overline{(hh')}\rightarrow \overline{max_\sim \emph{BI}_h}))$.  For b), there are two cases: b1) $\neg (\overline{(hh')}\rightarrow \overline{\emph{max}_>Z_h\backslash \emph{BI}_h})$, b2) $\neg (\overline{(hh')}\rightarrow \overline{max_\sim \emph{BI}_h}))$. Both of the two cases would tell us: $(hh')\notin \textbf{BI}_h$. Then by the definition of \textbf{SCBI},  we know $h^*\notin \textbf{SCBI}$. The same contraction with 1).

2) can be proved in a similar style.

%For 2). Suppose otherwise, Then $\exists h^*$ $\exists (hh')\lhd h^*$ with a) $A(\overline{h^*}\rightarrow \overline{\rm{BI}})$, and b) $not$ $(\overline{(hh')}\rightarrow \overline{\emph{max}_{\succeq}Z_h})$.  From (b), it follows that $(hh')\notin \emph{max}_{\succeq}Z_h$. Then by the definition of \textbf{SCBI},  we get $h^*\notin \textbf{SCBI}$. But then, by a), $h^*\in \textbf{BI}$. Contradiction.

\smallskip

Direction ($\Leftarrow$). Suppose that $\neg(A(\overline{\rm{BI}}\leftrightarrow \overline{\rm{SCBI}}))$. Then

\smallskip

$(a)$: $\exists z^*\in \textbf{BI}$ and $z^*\notin \textbf{SCBI}$, or

\smallskip

$(b):$ $\exists z^*\in \textbf{SCBI}$ and $z^*\notin \textbf{BI}$.

\smallskip

\noindent If $(a)$, then, by the antecedent, we have that: $\forall (ha)\lhd z^*, (ha)\in H_h$.
Also, $\forall (hh')\in Z_h$ and $(hh')\lhd h^*$, it holds that $(hh')\in \emph{max}_{\succeq}Z_h$. Then it directly follows that $z^*$ is a SCBI history. Contradiction.

\smallskip

If $(b)$, then take any $z\in \textbf{BI}$, which shares a prefix $u$ with $z^*$, i.e., $u\lhd z$ and $u\lhd z^*$.  By the antecedent, we have $z_u\in \emph{max}_{\succeq}Z_h$. Since $z^*\notin \textbf{BI}$, it follows that $z_u> z^*_u$. Then $z^*\notin \textbf{SCBI}$. Once more, we have a contradiction. $\Box$
\end{proof}

%%%%%%%%%%%%%%%%%%%%%%%%%%%%%%%%%%%%%%第3.2节%%%%%%%%%%%%%%%%%%%%%%%%%%%%%%%%%%%%%%%%%%%%%%%%%%%%%%%%%%%%%%%%%%%%%%%%
\subsection{Valid principles}

The operator $[!\varphi]$ makes $\mathcal{L}$ a PAL-like language. However, the special model-update makes it different from standard PAL \cite{Ditmarsch2007}.  This suggests a close look at what is and what is not valid in preference-sight models.

First, some axioms in standard PAL do not hold in preference-sight models. For example, the $!\emph{ATOM}$ axiom,  $[!\varphi]p\leftrightarrow (\varphi\rightarrow p)$, is not valid when it is of the form below.

\begin{proposition} The following is not valid in preference-sight models, where $h,h_1,h_2$ represent arbitrary histories.

\smallskip

$!\emph{Sight-Preference}:$  \hspace{2ex} $[!\overline{s(h)}]\overline{h_1\geq h_2}\leftrightarrow (\overline{s(h)}\rightarrow \overline{h_1\geq h_2})$.
\end{proposition}

\begin{proof} For a counterexample, consider the tree $T$ in Figure \ref{Figure:preferenceandsight}. It is easy to see that in the model $M$ for $T$, $M\models [!\overline{s(\varepsilon)}]\overline{h_1\geq h_2}$ and $M\nvDash \overline{s(\varepsilon)}\rightarrow \overline{h_1\geq h_2}$, since there exists a state $\varepsilon$ such that  $M, \varepsilon \models \overline{s(\varepsilon)}$ and $M, \varepsilon \nvDash\overline{h_1\geq h_2}$.
\end{proof}

\vspace{-1ex}

This proposition says that subjective preference in visible trees is not necessarily consistent with objective preference.

Now let us see some interesting valid principles and their intuitive interpretations.

\begin{lemma} The formulas shown in Table \ref{table:validity} are valid, where $h, h_1, h_2,$ and $h_3$ are arbitrary histories. \end{lemma}

\begin{table}[h]
\begin{center}
\fontsize{10pt}{1.1\baselineskip}\selectfont{\begin{tabular}{|l|l|}
\hline
$\emph{Taut}$ & all propositional tautologies\\
\hline

\vspace{-1.5ex}

    &  \\

$T_{\geq}$    &  $\overline{h\geq h}$\\

%\hline

$4_{\geq}$        & $\overline{h_1\geq h_2}\wedge \overline{h_2\geq h_3} \rightarrow \overline{h_1\geq h_3}$\\

%\hline

$to_{\geq}$   &  $\overline{h_1\geq h_2} \vee \overline{h_1\geq h_2}$\\

\hline

\vspace{-1.5ex}

    &  \\

$T_{s}$        &  $\overline{h}\rightarrow \overline{s(h)}$\\

$\emph{TM}$  & $\bigwedge\limits_{z\in Z}\bigwedge\limits_{h}(A(\overline{z}\rightarrow \overline{h})\rightarrow A(\overline{h}\rightarrow \overline{z}))$\\

$\emph{DC}$  & $\bigwedge\limits_{h_3}\bigwedge\limits_{h_2\lhd h_3} \bigwedge\limits_{h_1\lhd h_2}(A(\overline{h_3}\rightarrow \overline{s(h_1)})\rightarrow A(\overline{h_2}\rightarrow \overline{s(h_1)}))$\\

$\emph{NF}$  & $\bigwedge\limits_{h_3}\bigwedge\limits_{h_2\lhd h_3} \bigwedge\limits_{h_1\lhd h_2}(A(\overline{h_3}\rightarrow \overline{s(h_1)})\rightarrow A(\overline{h_3}\rightarrow \overline{s(h_2)}))$\\

\hline

$!\tiny{\emph{ATOM$\setminus$SP}}$ &  $[!\varphi]p\leftrightarrow (\varphi\rightarrow p)$ \\
&(excluding the schema $!\emph{Sight-Preference}$)\\

$!\emph{NEG}$ &  $[!\varphi]\neg\psi\leftrightarrow (\varphi\rightarrow \neg[\varphi]\psi)$\\

$!\emph{CON}$ &  $[!\varphi](\psi\wedge \chi)\leftrightarrow ([!\varphi]\psi\wedge[!\varphi]\chi)$\\

$!\emph{COM}$ &  $[!\varphi][!\psi]\chi \leftrightarrow ![\varphi\wedge[!\varphi]\psi]\chi$\\

$\emph{Dual}$ & $[!\varphi]\psi\leftrightarrow\neg\langle!\varphi\rangle\neg\psi$\\

\hline

\end{tabular}
} \vspace{-0.5ex}\caption{Valid principles of $L$\vspace{-2ex}}\label{table:validity}
\vspace{-2ex}\end{center}%\vspace{-3ex}
\end{table}
\vspace{-1ex}

\begin{proof}  We only prove some cases, proofs for the others are trivial or standard.

For $T_s$. Take any state $u$ with $M,u\models \overline{h}$. Then $u\in \mathcal{V}(\overline{h})$. As the sight function is reflexive, i.e., $h\in s(h)$, it holds that $\mathcal{V}(\overline{h})\subseteq \mathcal{V}(\overline{s(h)})$. So $u\in \mathcal{V}(\overline{s(h)})$.  Thus,  $M,u\models\overline{s(h)}$.

For $TM$. Take any state $u$, any history $h$ and any $z\in Z$, and suppose $M,u\models A(\overline{z}\rightarrow \overline{h})$. Then for any $u'$, $u'\in \mathcal{V}(\overline{z})$ implies that $u'\in \mathcal{V}(\overline{h})$. Thus, $\mathcal{V}(\overline{z})\subseteq \mathcal{V}(\overline{h})$. It follows that $z\in \mathcal{V}(\overline{h})$. Given that $z$ is terminal, by the definition of $ \mathcal{V}(\overline{h})$, it must be that $h=z$.  Thus, $M,u\models A(\overline{h}\rightarrow \overline{z})$.

For $\emph{DC}$. Take  any state $u$, suppose for some $h_1\lhd h_2\lhd h_3$, $M,u\models A(\overline{h_3}\rightarrow \overline{s(h_1)})$. Then we know $\mathcal{V}(\overline{h_3})\subseteq \mathcal{V}(\overline{s(h_1)})$. It follows that $h_3\in s(h_1)$.  As the sight function is downward closed, we have $h_2\in s(h_1)$. Thus, $M,u\models A(\overline{h_2}\rightarrow \overline{s(h_1)})$.

For $!\tiny{\emph{ATOM$\setminus$SP}}$. Take any state $u$, and let $M,u\models [!\varphi]p$ where $\varphi$ is not of the form $!\overline{s(h)}$ and $p$ is not of the form $\overline{h_1\geq h_2}$.  It holds that $M,u\models \varphi$ implies that $M_{!\varphi},u\models p$. By Definition \ref{modelupdate}, $M_{!\varphi},u\models p$ iff $M,u\models p$. Therefore, $M,u\models \varphi$ implies $M,u\models p$. Equivalently, then, $M,u\models \varphi\rightarrow p$. $\Box$\end{proof}

{\textbf{Interpretation of valid principles.}} Each of these axioms has some intuitive appeal. $T_{\geq}$, $4$ and $to_{\geq}$ show the \emph{reflexivity}, \emph{transitivity} and \emph{totality} of the preference relation, respectively.  Likewise, $T_{s}$ says that sight is \emph{reflexive}. $DC$ characterizes the (\emph{downward-closure}) property of sight. $NF$ encodes the \emph{non-forgetting} property of sight. $\emph{TM}$ guarantees that terminal histories of the P-S tree are actually \emph{terminal}. One further interesting point is that there is no correspondence of $\emph{TM}$ for  terminal histories of visible trees.

\begin{fact} The following formula is not valid in preference-sight models:

\hspace{13ex} $\bigwedge_u\bigwedge_{z\in Z_u}\bigwedge_{h}(A(\overline{z}\rightarrow \overline{h})\rightarrow  A(\overline{h}\rightarrow \overline{z})).$
\end{fact}

\vspace{1ex}

Other validities in the table are axioms for standard PAL. We postpone the study of a complete axiomatization of the logic L until future work.

To conclude this section, in $\mathcal{L}$, the ingredients including histories, preferences and sights are encoded as primitive propositions. Various earlier phenomena in P-S trees can thus be captured in a simple, direct and intuitive manner. This special-purpose logic, as we will see soon, is model-dependent, but it can also be formulated generically.

%%%%%%%%%%%%%%%%%%%%%%%%%%%%%%%%%%%%%%%%%%%%%%%%%%%%%%%%%%%%%%%%%%%%%%%%%%%%%%%%%%%%%%%%%%%%%%%%%%%%%%%%%%%%%%%%%%%%%%%%%%%%%%%%%%%%%%

\section{Background in  game logics}

In this section, we relate our logic $L$ to existing logics for classical game theory, showing how ideas can be combined where useful.  Since so far we have been working with BI and SCBI histories, we first define strategies for P-S trees: A \emph{strategy} for a P-S tree $(T,s)$ is a function $\sigma: H\rightarrow A$ such that $\sigma(h)\in A(h)$. That is, $\sigma$ assigns each history $h$ an action that follows $h$. In particular, for a visible tree $T_h$, a `local strategy' $\sigma_h$ is a restriction of $\sigma$ to $T_h$, such that $\sigma_h(h')=\sigma(h')$ for any $h'\in T_h$.

\subsection{Generic formulation of $\mathcal{L}$}

In applied logic for structure analysis, there exist two extremes, viz. model-dependent `local languages' and `generic languages' that work across models. For a generic logic, a definition of a property $\pi$ is a formula $\varphi$ such that for all models $M$, $M$ has property $\pi$ iff $M\models \varphi$. For a local language, such a formula can depend on a given model $M$: there exists a formula $\varphi_M$ which depends on $M$, such that any model $M$ has the property $\pi$ iff $M\models \varphi_M$. However, in this case, the defining formula can be trivial. For example, one might define $\varphi_M$ simply as follows.

\vspace{-3ex}

\begin{eqnarray*}
\varphi_M =
\begin{cases}
\top,       & \emph{if} ~ M ~\emph{satisfies} ~ \pi \\
\bot, & Otherwise
\end{cases}
\end{eqnarray*}
\vspace{-1ex}

In this subsection, using a well-known \emph{Rationality} property as an example,  we discuss how model-dependent our earlier language $L$ is, and then show how it can be formulated in a generic way.  We first recall the results on classical BI. Given that we have been dealing with single-agent cases until now, in this Section, we will adapt the results from the literature on multi-player games  to the single-player case.

\vspace{0.5ex}

The BI strategy \cite{vanBenthem2011,Benthem14} is  the largest subrelation $\sigma$ of the total \emph{move} relation that has at least one successor at each node, while satisfying the rationality (RAT) property:

\smallskip

\textbf{RAT} \hspace{2ex} No alternative move for the player yields an outcome via further play with $\sigma$ that is strictly better than all the outcomes resulting from starting at the current move and then playing $\sigma$ all the way down the tree.

\smallskip

As argued in \cite{vanBenthem2011,Benthem14}, this rationality assumption is a confluence property for action and preference:

\vspace{1ex}

\textbf{CF} \vspace{-3ex} $$\forall x\forall y(x\sigma y\rightarrow\forall z(x ~\emph{move} ~z \rightarrow$$

\vspace{-4ex}

 $$\exists u(end(u)\wedge y\sigma^*u \wedge \forall v((end(v)\wedge z\sigma^*v )\rightarrow u \geq v))))$$

\vspace{1ex}

We can observe that there is also a corresponding rationality property for the local BI strategies that constitute an SCBI, which should however now express a confluence property for action, preference and sight. Specifically, for a P-S tree, each local BI strategy for the visible tree $T_h$ at $h$ is the largest subrelation $\sigma_h$ of the total \emph{move} relation in $T_h$, satisfying 1) $\sigma_h$ has at least one successor at each $h'\in T_h$, and 2) the following rationality property holds:

\smallskip

\textbf{RATS} \hspace{1.6ex} In the visible tree, there is one outcome obtained by playing $\sigma_h$ from the start to the end, that is no worse than all the outcomes yielded from any alternative first move followed by further play with $\sigma_h$.

\smallskip

This confluence property involving sight is expressible as follows in our language $\mathcal{L}$:

\vspace{-0.5ex}

\begin{proposition} Let $(T,s)$ be a P-S tree, and let $M$ be any model for it. $M$ satisfies \emph{RATS}  iff $M$ validates the following $\mathcal{L}$-formula, where $\sigma_h$ is the \emph{BI} strategy for visible tree at $h$ and where $(h(\sigma_h)^k)$ stands for the history reached from $h$ after executing $\sigma_h$ for $k$ times.
\end{proposition}

${{\textbf{CFS}}}_{M}$ \vspace{-3.5ex}  $$\bigwedge\limits_{h}\bigvee\limits_{z\in Z_h}\bigvee\limits_{k=l(z)-l(h)}(A(\overline{(h(\sigma_h)^k)} \leftrightarrow \overline{z})$$

\vspace{-3ex}

$$\rightarrow (\bigwedge\limits_{a'\in A_h(h)}\bigwedge\limits_{z'\in Z_h}\bigwedge\limits_{m=l(z')-l(ha')}(A(\overline{(ha'(\sigma_{h})^m)} \leftrightarrow \overline{z'}))\rightarrow$$

\vspace{-2ex}

$$\overline{z\geq z'})),$$

\begin{proof} We first claim that at state $h\in H$, for any terminal history $z\in Z_h$, and $h'\in H_h$, $A(\overline{h'} \leftrightarrow\overline{z})$ implies that $h'=z$. This is straightforward since $A(\overline{h'} \leftrightarrow\overline{z})$ demonstrates that prefixes of $h'$ are  the same with those of $z$, which means that $h'=z$. Then $M\models \textbf{CFS}_M$ says that there is a terminal history $z_h$ following $h$ by playing a local BI strategy $\sigma_h$, such that $z\succeq  z'$ for any other  $z'\in Z_h$ which follows an alternative first move $a'\in A_h(h)$ via further play of $\sigma_h$. Therefore, we know that $M$ satisfies \textbf{RATS}. $\Box$\end{proof}

\vspace{-1ex}

However, compared with the generic logic in \cite{BenthemG10,Benthem14,vanBenthem2011}, the given definition in our logic is local. It is obvious that $\textbf{CF}$, the formula defining the property \textbf{RAT}, is insensitive to models --  while our \textbf{CFS}$_M$ relies on a given model for its ranges of big disjunctions and conjunctions, and in its model-dependent notations like $s(h)$ and $\overline{h_1\geq h_2}$.  Still, it is also clearly true that our definition is not as trivial as the earlier local trick. Therefore, our logic $\mathcal{L}$ seems somewhere between the two extremes of locality and genericity.  This feeling can be made precise by moving to a closely related truly generic first-order logic.

The relevant modified formula involves some natural auxiliary predicates. $x\lhd y$ says that $x$ is a prefix of $y$; $x \,\sphericalangle \, y$ means that $x$ can see $y$. Corresponding to the BI relation $\sigma$, $y \sigma(x) z$ says that from $y$, $z$ is a local backward induction move in the visible tree at $x$; $\sigma^k$ describes $\sigma$ being  composed  for $k$ times with $k\in \mathbb{N}$ \footnote{Here $x\sigma^k y$ is the abbreviation of $\exists y_1\exists y_2 \cdots\exists y_k(x\sigma y_1\wedge y_1\sigma y_2\wedge\cdots \wedge y_{k-1}\sigma y_k \wedge (y_k=y)).$}; $move$ and $\geq$ are still the move relation and preference relation, respectively, of the game.

%\vspace{-1ex}

\begin{proposition} Any model $M$ satisfies \textbf{\emph{RATS}} iff it validates the following formula.
\end{proposition}

$\textbf{CFS}{(FO)}$:

\vspace{-5ex}

$$\forall x\{(\exists y(x\lhd y))\rightarrow$$

\vspace{-3ex}

$$\forall u[(x\sigma(x) u)\rightarrow \forall t( (x ~\emph{move}~t \wedge x\sphericalangle t)\rightarrow $$\\

\vspace{-9ex}

$$\exists z((x\sphericalangle z\wedge \neg\exists z'(z\lhd z' \wedge x\sphericalangle z')\wedge \exists k(u(\sigma(x))^k z))\wedge $$\\

\vspace{-9ex}

$$\forall v((x\sphericalangle v\wedge \neg\exists v'(v\lhd v' \wedge v\sphericalangle v')\wedge \exists l(t({\sigma(x))}^l v))\rightarrow $$\\

\vspace{-9ex}

$$\wedge z\geq v)))]\}.$$

\begin{proof} It is easy to show that $M\models \textbf{CFS}{(FO)} \emph{ iff } M\models \textbf{CFS}_M.\Box$
\end{proof}

\vspace{-1ex}

In summary, incorporating basic elements of P-S trees directly into first-order syntax makes $L$ intuitive and natural.

\smallskip

Even so, other logics exist for dealing with further aspects of game trees and solution procedures, and we will discuss a few  examples in what follows with a view to how they behave in the presence of sight.

\subsection{Solution procedures and fixed-point logics}

Recursive solution procedures naturally correspond to definitions in existing fixed-point logics, such as the widely used system  LFP(FO). An LFP(FO) formula mirroring the recursive nature of BI is constructed in \cite{BenthemG10,BenthemPR2011} to define the classical BI relation, based on the above property \textbf{RAT}.  Now, we have shown that sight-restricted SCBI, too, is a recursive game solution procedure.  Can LFP(FO) be used to define SCBI as well -- and if so, how?

%$bi::=$ $\forall x\forall y((\emph{Turn}_i(x)\wedge x\sigma y)\rightarrow \forall z(x ~\emph{move} ~z \rightarrow \exists u\exists v(end(u)\wedge end(v)\wedge y\sigma^*u \wedge z\sigma^*v \wedge u \geq_i v)))$

The answer is yes, but we need an extension. Rather than a binary relation $bi$ as in \cite{BenthemG10,BenthemPR2011}, characterizing SCBI needs a \emph{ternary} relation.  First, we define the local BI relation in visible trees, which will be denoted by $bi_{\emph{sight}}$.  For any states $x,y,z$,  $bi_{\emph{sight}}(x,y,z)$ means that in the visible tree at $x$, the local BI strategy is $bi_{\emph{sight}}$, which chooses $z$ when the current state is $y$. It is then obvious that $bi_{sight}$ should satisfy the following simple first-order definable property, requiring the relevant states to be visible and reachable:

\vspace{-2ex}

$$bi_{sight}(x,y,z)\rightarrow \emph{see}(x,y)\wedge \emph{see}(x,z)\wedge \emph{move}(y,z).$$

The intuition of $bi_{\emph{sight}}(x,y,z)$ is then captured as follows:

\vspace{-2ex}

$$\forall x \forall y \forall z(bi_{\emph{sight}}(x,y,z)\rightarrow \forall t((see(x,t) \wedge \emph{move}(y,t))$$

\vspace{-5ex}

$$\rightarrow(\exists u (end_{\emph{sight}}(x,u)\wedge bi_{\emph{sight}}^*(x,z,u)\wedge \forall v((end_{\emph{sight}}(x,v)  \wedge$$

\vspace{-5ex}

$$bi_{\emph{sight}}^*(x,t,v))\rightarrow u\geq v))))).$$

Notice that all occurrences of  $bi_{\emph{sight}}$ in the above formulas are still syntactically positive. This allows us to define local BI strategy $bi_{\emph{sight}}$ with LFP(FO).

\vspace{-1ex}

\begin{proposition} The strategy $bi_{sight}$ can be defined as the relation $R$ in the following \emph{LFP(FO)} formula.

\vspace{-2ex}

$$\nu R, xyz\bullet \forall x \forall y \forall z(R(x,y,z)\rightarrow \forall t((see(x,t) \wedge \textit{move}(y,t))$$

\vspace{-5ex}

$$\rightarrow(\exists u (end_{\textit{sight}}(x,u)\wedge R^*(x,z,u) \wedge \forall v( (end_{\textit{sight}}(x,v)\wedge$$

\vspace{-5ex}

$$R^*(x,t,v))\rightarrow u\geq v))))).$$
\end{proposition}

\vspace{-1ex}

It can be proved formally that $bi_{\emph{sight}}$ is a greatest-fixed-point of the above formula. Based on  $bi_{\emph{sight}}$, we now proceed to show that the SCBI relation is LFP(FO) definable.

\vspace{-1ex}

\begin{corollary} The \emph{SCBI} relation \textit{scbi} for a P-S tree can be represented in the following formula:

\vspace{-3ex}

$$\forall x \forall y(scbi(x,y)\leftrightarrow bi_{\emph{sight}}(x,x,y)).$$
\end{corollary}

As in the original classical case, this LFP(FO) definability of \emph{scbi} exposes an intersection between the logical foundation of computation and the recursive nature of sight-compatible backward induction solutions for P-S trees.

%%%%%%%%%%%%%%%%%%%%%%%%%%%%%%%%%%%%%%%%%%%%%%%%%%%%%%%%%%%%%%%%%%%%%%%%%%%%
\subsection{Modal surface logic of best action}

In contrast with detailed formalism of solutions with LFP(FO), there is the modal surface logic of \cite{BenthemOtterlooRoy}, which enables direct and natural reasoning about best actions without considering the underlying details of recursive computation. First of all, we list its modalities for classical BI. $[bi]$ and $[BI]$ encode the BI move and BI paths respectively. $[\emph{best}]\varphi$ says that $\varphi$ is true in some successor of the current node that can be reached in one step via the \emph{bi} move.

\medskip

$M,h\models end$ iff $h\in Z$.

\vspace{0.5ex}

$M,h\models [move]\varphi$ iff $\forall$ $h'=(ha)$ with $a\in A(h)$, $M, h'\models \varphi$.

\vspace{0.5ex}

$M,h\models [best]\varphi$ iff for all $h'$ with  $h'\in bi(h)$, $M, h'\models \varphi$.

\vspace{0.5ex}

$M,h\models [bi]\varphi$ iff for all $h'$ with  $h'\in bi(h)$, $M, h'\models \varphi$.

\vspace{0.5ex}

$M,h\models [bi^*]\varphi$ iff $M, u\models \varphi$ for all $u$ with $u\in (bi)^*(h)$.

\vspace{0.5ex}

$M,h\models [BI]\varphi$ iff for all $z$ with  $z\in \textbf{BI}$, $M, z\models \varphi$.

\vspace{2ex}

The above logic is still applicable in our setting, but it requires substantial extension for sight-related concepts. In accordance with $[bi]$ and $[BI]$, we use $[scbi]$ and $[\emph{SCBI}]$ as  operators for the SCBI strategy and SCBI path, respectively.  For the local BI strategy and path in visible trees, the modalities are  $[bi_{\emph{sight}}]$ and $[\emph{BI}_{\emph{sight}}]$. Moreover, recall that $M_{!s(h)}$ is the updated model obtained in the way of Definition \ref{modelupdate}.

\medskip

$M,h \models [scbi]\varphi$ iff for all $h'$ with $h'\in scbi(h)$, $M, h'\models \varphi$.

\vspace{0.5ex}

$M,h\models [\emph{SCBI}]\varphi$ iff for all $h'$ with  $z\in \textbf{SCBI}$, $M, z\models \varphi$.

\vspace{0.5ex}

$M,h\models [!\emph{sight}]\varphi$ iff $ ~M_{!s(h)},h\models \varphi.$

\vspace{0.5ex}

$M_{!s(h)}, u\models end_{\emph{sight}}$ iff $u\in Z_h.$

\vspace{0.5ex}

$M_{!s(h)}, u\models [move_{\emph{sight}}]\varphi~$ iff  for $\forall u'=(ua)$ $\emph{with} ~ a\in A_h(u),$

\vspace{-0.2ex}

\hspace{28ex} $M_{!s(h)}, u'\models \varphi$.

\vspace{0.2ex}

$M_{!s(h)}, u \models [best_{\emph{sight}}]\varphi$ iff $M,u'\models \varphi$ for $\forall u'\in bi_h(u)$.

\vspace{0.5ex}

$M_{!s(h)}, u \models [bi_{\emph{sight}}]\varphi$ iff $M, u'\models \varphi$ for $\forall u'\in bi_h(u)$.

\vspace{0.5ex}

$M_{!s(h)}, u\models [(bi_{\emph{sight}})^*]\varphi$  iff $M_{!s(h)}, u'\models \varphi$ for all $u'$,

\hspace{27.5ex}                                   such that $u'\in (bi_h)^*(u)$.

\vspace{0.5ex}

$M,h\models [BI_\emph{sight}]\varphi$ iff for all $z$ with  $z\in \textbf{BI}_h$, $M, z\models \varphi$.

%%%%%%%%%%%%%%%%%%%%%%%%%%%%%%%%%%%%%%%%%%%%%%%%%%%%%%%%%%%%%%%%
\medskip

We give a few illustrations of  new issues that arise now.

\smallskip

\noindent{\bf Capturing the SCBI strategy} For a start, we are now able to characterize the SCBI strategy, in a similar vein as the frame correspondence for the classical BI strategy in \cite{BenthemOtterlooRoy}.

\vspace{-1ex}

\begin{proposition} The \emph{BI} strategy is the unique relation $bi$ satisfying this modal axiom for all propositions $p$:

\vspace{-3ex}

$$(\langle \textit{bi}^*\rangle(\textit{end}\wedge p))\rightarrow([\textit{move}][\sigma^*](\textit{end}\wedge \langle\leq\rangle p))$$
\end{proposition}

\vspace{-1ex}

Along the same lines,  we can express the SCBI strategy in P-S trees based on the idea that each \emph{scbi} move coincides with a local BI move within the current visible tree.

\vspace{-1ex}

\begin{proposition}  The \emph{SCBI} strategy is the  relation $scbi$ satisfying the following axioms for all propositions $p$:
\end{proposition}

\vspace{-4ex}

$$(1)~~~ \langle scbi\rangle p\leftrightarrow [!\emph{sight}]\langle\emph{bi}_{\emph{sight}}\rangle p. $$

\vspace{-3ex}

$$(2) ~~~[!\emph{sight}](\langle (\emph{bi}_{\emph{sight}})^*\rangle (\emph{end}_{\emph{sight}} \wedge p) \rightarrow$$

\vspace{-4ex}

$$[\emph{move}_{\emph{sight}}]\langle (\emph{bi}_\emph{sight})^*\rangle(\emph{end}_{\emph{sight}}\wedge \langle \leq\rangle p)).$$

\noindent{\bf Best action and preference-consistency} Turning to properties of frames for the extended modal logic of best action with sight, there are  interesting differences  when comparing SCBI and classical BI. To see this, we employ operators  $\langle \emph{best} \rangle$, $\langle \emph{best}_\emph{sight}\rangle$, $\langle\textit{bi}^*\rangle$, $\langle\textit{scbi}^*\rangle$ and $(bi_{\textit{sight}})^*$. Now we can make some interesting comparisons.

\vspace{-1ex}

\begin{proposition} For classical backward induction, the axiom
$\langle best \rangle \langle \textit{bi}^*\rangle\varphi\leftrightarrow \langle \textit{bi}^*\rangle\varphi$ holds.
\end{proposition}

However, the new frames do not have the corresponding axiom for the SCBI strategy, since the actions it recommends are not necessarily the actual best actions according to BI. Even in visible trees, this is also not true.

\vspace{-0.5ex}

\begin{proposition} The following formulas are not valid:

\vspace{-3ex}

$$(a) ~~\langle best \rangle \langle \textit{scbi}^*\rangle\varphi\leftrightarrow \langle \textit{scbi}^*\rangle\varphi.$$

\vspace{-4ex}

$$(b) ~~[!\textit{sight}] (\langle best \rangle \langle (bi_{\textit{sight}})^*\rangle\varphi\leftrightarrow \langle (bi_{\textit{sight}})^*\rangle\varphi).$$

\end{proposition}

Nevertheless, there is a certain coherence between the local BI strategy and local best actions returned by it.

\vspace{-1ex}

\begin{proposition} The following formula is valid: $[!\textit{sight}](\langle best_{\textit{sight}} \rangle \langle (\textit{bi}_{\textit{sight}})^*\rangle\varphi\leftrightarrow \langle (\textit{bi}_{\textit{sight}})^*\rangle\varphi).$
\end{proposition}

As for the preference relation, SCBI has a property that classical BI lacks: local BI moves never conflict with the preferences in submodels. In other words, within a visible tree, the initial move determined by the local BI strategy is more preferable for the agent than any other first move.

\begin{proposition} For \emph{SCBI}, it holds that $[!\textit{sight}](\langle best_{\textit{sight}}\rangle\varphi \rightarrow [move_\textit{sight}]\langle\leq\rangle\varphi).$
\end{proposition}

For  BI, although it returns a final optimal path, there is no guarantee that its intermediate histories  be preferable.

\vspace{-1ex}

\begin{proposition} For \emph{BI}, the following does not hold: $\langle best \rangle\varphi \rightarrow [move]\langle\leq \rangle\varphi.$
\end{proposition}

\noindent{\bf Path terminality and optimality} Using a similar style of modal analysis, we can make the following observations concerning the obvious operators  $[\emph{BI}]$, $[\emph{SCBI}]$ and $[\emph{BI}]_{sight}$.
%In standard extensive games, $[\emph{BI}] \varphi$  means that $\varphi$ is true at any terminal history which is a BI path from the current state.

\begin{proposition}

We have the following three facts:
\vspace{-0.5ex}

\begin{itemize}
\item[$(a)$] The formula$ [\textit{BI}]\varphi \rightarrow [\textit{BI}][\textit{BI}]\varphi$ is valid.

\item[$(b)$] For SCBI, the following formula does not hold: $[BI_{\textit{sight}}]\varphi \rightarrow [\textit{BI}_{\textit{sight}}][\textit{BI}_{\textit{sight}}]\varphi.$

\item[$(c)$]  The formula $ [\textit{SCBI}]\varphi \rightarrow [\textit{SCBI}][\textit{SCBI}]\varphi$ is valid.

\end{itemize}
\end{proposition}

Here (a) says that from a BI outcome only a terminal history can be reached; (b) shows that the local BI history may not be a terminal history of the whole tree, and (c) says  the SCBI history for the whole tree is always terminal.

Another phenomenon  regarding these operators is the local optimality of SCBI at the cost of being more realistic than BI. We have mentioned this point already in Section 2.2.4: now we can present a precise formal version.

\vspace{-1ex}

\begin{proposition} Let $\sigma$ be any strategy profile,
\begin{itemize}
\item[$(a).$] For BI, the following is valid: $\langle BI\rangle\varphi\rightarrow [\sigma]\langle\leq\rangle\varphi.$

\item[$(b).$] The following does not hold: $\langle \textit{SCBI}\rangle\varphi\rightarrow [\sigma]\langle\leq\rangle\varphi.$

\item[$(c).$]  For SCBI, it holds that $[!\textit{sight}](\langle BI_{\textit{sight}}\rangle\varphi\rightarrow [\sigma_{\textit{sight}}]\langle\leq_{\textit{sight}}\rangle\varphi).$
\end{itemize}
\end{proposition}

Here $(a)$ shows the global optimality of the BI path. $(b)$ and $(c)$ together say the SCBI path is not globally optimal, but each move on this path leads to a locally optimal path.

\smallskip

Altogether, this section has shown the broad logical foundations of our framework, embedding our local language in existing broader generic formalisms, but also enriching and extending these frameworks with aspects of short sight.

\section{Toward Multi-player games}

While our models and results are about single-agent sequential decision-making processes, we believe they are applicable well beyond that. They can be naturally extended to multi-player extensive game-scenarios with short sight. For such a game model, we can build on \cite{GrossiT12:Sight}, which makes an assumption that the current player only knows his own sight, and that he believes other players can see as much as he can see and will play according to this belief. That is, this model precludes more complex forms of interactive knowledge and reasoning. But using this same assumption, our model in this paper can be extended to multi-player cases directly. The only thing we have to do is add agent-labeling to SSDM: even though players can  change with time, everything including sight, preference, and actions can be modeled from the current player's perspective.

We will not state any results for the extended multi-player model since they are quite similar to what we have shown already. The case where we drop the above assumption and  allow a more free modeling of players' mutual knowledge and beliefs about sight and preference would be more interesting. We will leave this for future work.

\section{Discussion and Conclusion}

Though motivated  by single-agent decision-making process, we have gone towards a much more general goal In the process, our analysis significantly adds to  current connections between logic, computation, and game solutions.

In many recent game-theoretic papers centering on \emph{bounded rationality}, a model has been used of \emph{games with awareness}, \cite{Halpern2011,HalpernR06,HeifetzMS13,Feinberg,Halpern14}.  This approach generalizes the classical representation of extensive games by modeling  players who may not be aware of all the paths. While \cite{GrossiT12:Sight} shows that games with short sight are a well-behaved subclass of games with awareness, there exists a fundamental difference in focus. Players in the latter approach may be unaware of some branches but they can always see some terminal histories, while in the former, players' sight may only include intermediate histories, ruling out all  terminal ones. Moreover, we have shown how short-sight games allow for a natural co-existence of two views of a game, that of insiders and that of outsiders. Having said this, it is clearly an interesting issue to see if our approach in this paper can be extended to cover awareness.

Another obvious interface for our logics are heuristic evaluation approaches for intermediate nodes  used by  the AI community for computational game-solving, \cite{Lim2006,edwards:mittr63,2011Rossi}. This, too, is a connection that deserves further exploration.

There are many additional topics to pursue. We already mentioned multi-player scenarios with non-trivial interactive reasoning about other agents' preferences, sights, and strategies. This has also been identified as a key  task for epistemic game theory, \cite{perea2014}.

\section*{Acknowledgments} I thank Fenrong Liu for our fruitful collaboration on earlier versions of this paper. Paolo Turrini provided crucial insights on short-sight games and their connections with games and
computation, which we are partly exploring together. Sonja Smets provided helpful comments overall. But especially, I thank Johan van Benthem for our longstanding contacts on the logic of short-sight games: Section 4 of this paper owes a lot to his many suggestions and observations. This work is supported by the China Scholarship Council and NSFC grant No.61472369.
 %The author acknowledges the financial support received from the China Scholarship Council.

%
% The following two commands are all you need in the
% initial runs of your .tex file to
% produce the bibliography for the citations in your paper.
\bibliographystyle{eptcs}
\bibliography{gss}  % sigproc.bib is the name of the Bibliography in this case
\end{document}